\documentclass[
submission
]{dmtcs-episciences}

\usepackage[utf8]{inputenc}
\usepackage{subfigure}
\usepackage{multirow}
\usepackage[round]{natbib}
\usepackage{amsmath,amscd,amssymb,amsfonts}
\usepackage{xcolor}
\newtheorem{definition}{Definition}
\newtheorem{observation}{Observation}
\newtheorem{lemma}{lemma}
\newtheorem{theorem}{Theorem}

\author{Susobhan Bandopadhyay\affiliationmark{1}
  \and Sasthi C. Ghosh\affiliationmark{2}
  \and Subhasis Koley\affiliationmark{2}}
\title[Improved Bounds on the Span of $L(1,2)$-edge Labeling of Some Infinite Regular Grids]{Improved Bounds on the Span of $L(1,2)$-edge Labeling of Some Infinite Regular Grids\footnote{Preliminary version of this paper was presented at the 18th Cologne-Twente Workshop on Graphs and Combinatorial Optimization (CTW 2020), September 14-16, 2020~\cite{ours} }}
\affiliation{
  School of Computer Sciences, National Institute of Science Education and Research,  Bhubaneswar, India\\
  Advanced Computing and Microelectronics Unit, Indian Statistical Institute, Kolkata, India}
\begin{document}
\maketitle
\begin{abstract}
  For two given nonnegative integers $h$ and $k$, an $L(h,k)$-edge labeling of a graph $G$ is the assignment of labels $\{0,1, \cdots, n\}$ to the edges so that two edges having a common vertex are labeled with difference at least $h$ and two edges not having any common vertex but having a common edge connecting them are labeled with difference at least $k$. The span $\lambda'_{h,k}{(G)}$ is the minimum $n$ such that $G$ admits an $L(h,k)$-edge labeling. Here our main focus is on finding $\lambda'_{1,2}{(G)}$ for $L(1,2)$-edge labeling of infinite regular hexagonal ($T_3$), square ($T_4$), triangular ($T_6$) and octagonal ($T_8$) grids. It was known that $7 \leq \lambda'_{1,2}{(T_3)} \leq 8$, $10 \leq \lambda'_{1,2}{(T_4)} \leq 11$, $16 \leq \lambda'_{1,2}{(T_6)} \leq 20$ and $25 \leq \lambda'_{1,2}{(T_8)} \leq 28$. Here we settle two long standing open questions i.e. $\lambda'_{1,2}{(T_3)}$ and $\lambda'_{1,2}{(T_4)}$. We show $\lambda'_{1,2}{(T_3)} =7$, $\lambda'_{1,2}{(T_4)}= 11$. We also improve the bound for $T_6$ and $T_8$ and prove  $\lambda'_{1,2}{(T_6)} \geq 18$, $ \lambda'_{1,2}{(T_8)} \geq 26$. 

\end{abstract}




\section{Introduction}\label{sec:1}

\textit{Channel assignment problem} (CAP) is one of the fundamental problems in wireless communication where frequency channels are assigned to transmitters such that interference can not occur. The objective of the CAP is to minimize the span of frequency spectrum. In 1980, Hale ~\cite{Hale} first formulated the CAP as a classical vertex coloring problem. Later on, in 1988 Roberts ~\cite{Roberts} introduced $L(h,k)$-vertex labeling as defined below:
\begin{definition}
For two non-negative integers $h$ and $k$, an $L(h,k)$-vertex labeling of a graph $G(V,E)$ is a function $\mathbf{f}:V \xrightarrow{}\{0,1,\cdots, n\}, \forall v \in V $ such that $\vert \mathbf{f}(u)-\mathbf{f}(v) \vert \geq h$ when $d(u,v)=1$ and $\vert \mathbf{f}(u)-\mathbf{f}(v) \vert \geq k$ when $d(u,v)=2$. Here, distance between vertices $u$ and $v$, $d(u,v)$ is $k'$ if at least $k'$ edges are required to connect $u$ and $v$.
\end{definition}
The \textit{span} $\lambda_{h,k}(G)$ of $L(h,k)$-vertex labeling is the minimum $n$ such that $G$ admits an $L(h,k)$-vertex labeling. In 1992 Griggs and Yeh ~\cite{GriggsYeh} extended the concept by introducing $L(k_1,k_2,\cdots, k_l)$-vertex labeling with separation $\{k_1,k_2,\cdots, k_l\}$ for $\{1,2,\cdots, l\}$ distant vertices and their main focus was on $L(h,k)$-vertex labeling for a special case $h=2,\ k=1$.  In 2007, Griggs and Jin ~\cite{GriggsJin} studied $L(h,k)$-edge labeling, which can be formally defined as: 
\begin{definition}
For two non-negative integers $h$ and $k$, an $L(h,k)$-edge labeling of a graph $G(V,E)$ is a function $\mathbf{f}:E \xrightarrow{}\{0,1,\cdots, n\}, \forall e \in E$ such that $\vert \mathbf{f}(e_1)-\mathbf{f}(e_2) \vert \geq h$ when $d(e_1,e_2)=1$ and $\vert \mathbf{f}(e_1)-\mathbf{f}(e_2) \vert \geq k$ when $d(e_1,e_2)=2$. Here, for any two edges $e_1$ and $e_2$, the distance $d(e_1,e_1)$ is $k'$ if at least $(k'-1)$ edges are required to connect $e_1$ and $e_2$.
\end{definition}
Like $L(h,k)$-vertex labeling, the \textit{span} $\lambda'_{h,k}(G)$ of $L(h,k)$-edge labeling is the minimum $n$ such that $G$ admits an $L(h,k)$-edge labeling.  In 2011, Calamoneri did a rigorous survey ~\cite{cala2} on both vertex and edge labeling problems. Authors in \cite{lin1,lin2,lin3,cala} have studied $L(h,k)$-edge labeling of regular infinite hexagonal ($T_3$), square ($T_4$), triangular ($T_6$) and octagonal ($T_8$) grids for the special case of $h=1$ and $k=2$. They obtained some upper and lower bounds on $\lambda'_{1,2}(G)$  for $T_3$, $T_4$, $T_6$ and $T_8$ with a gap between them. In this paper, we improved some of these gaps.

Given a graph $G(V,E)$, its \textit{line graph} $L(G)(V',E')$ is a graph such that each vertex of $L(G)$ represents an edge of $G$ and two vertices of $L(G)$ have an edge if and only if their corresponding edges share a common vertex in $G$. It is known that if $G$ is $d$-regular then $L(G)$ is $2(d-1)$-regular. Fig.~\ref{line_graph} shows $T_3$, $L(T_3)$, $L(T_4)$, $T_6$ and $T_8$. Note that
edge labeling of $G$ is equivalent to vertex labeling of $L(G)$. In our approach, instead of $L(1,2)$-edge labeling of
$T_3$ and $T_4$, we use $L(1,2)$-vertex labeling of $L(T_3)$ and $L(T_4)$. Note that, $L(T_6)$ and $L(T_8)$ are $10$-regular and $14$-regular, respectively. Because of this high degree, we consider $L(1,2)$-edge labeling of $T_6$ and $T_8$ directly. Our results on $\lambda'_{1,2}(G)$ for $T_3$, $T_4$, $T_6$, $T_8$ are stated in Table~\ref{tab1}. In this table, $a-b$ represents that $a \leq \lambda'_{1,2}(G) \leq b$. Here, we use coloring and labeling interchangeably.

\begin{figure}[ht]
\centerline{\includegraphics[width=10.0cm]{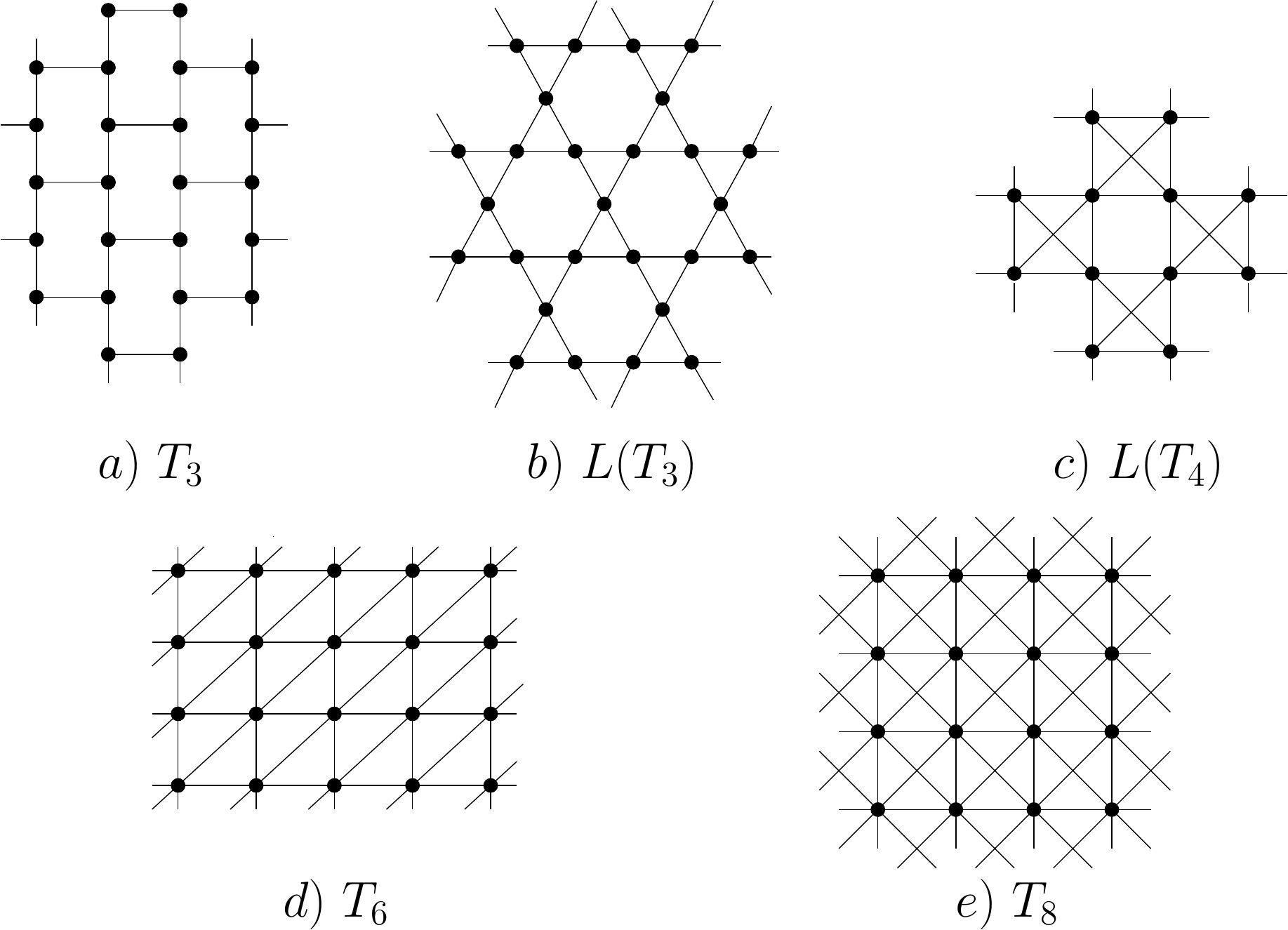}}
\caption{$T_3$, $L(T_3)$, $L(T_4)$, $T_6$ and $T_8$.}
\label{line_graph}
\end{figure}

\begin{table}[!ht]

\centering

    \caption{The main results.}     
    \label{tab1}

    \begin{small}
    \begin{tabular}{|c|c|c|c|c|c|c|c|c|}
    \hline
    \multirow{5}{*}{} &
      \multicolumn {4}{c|}{$\lambda'_{1,2}(G)$}  \\
    \hline
    \textbf{Grid} & $T_3$ & $T_4$ & $T_6$ & $T_8$ \\
    \hline
     \textbf{Known}     & 7-8 ~\cite{lin3}  & 10-11 ~\cite{lin3} & 16-20 ~\cite{cala} & 25-28 ~\cite{cala} \\
    \hline
     \textbf{Ours}   & 7    & 11  & 18-20 & 26-28\\
    \hline
    
    \end{tabular}
    \end{small} 
\end{table}

\section{Results}\label{sec:2}
 
\begin{figure}[!ht]
\begin{center}
\includegraphics[width=13cm]{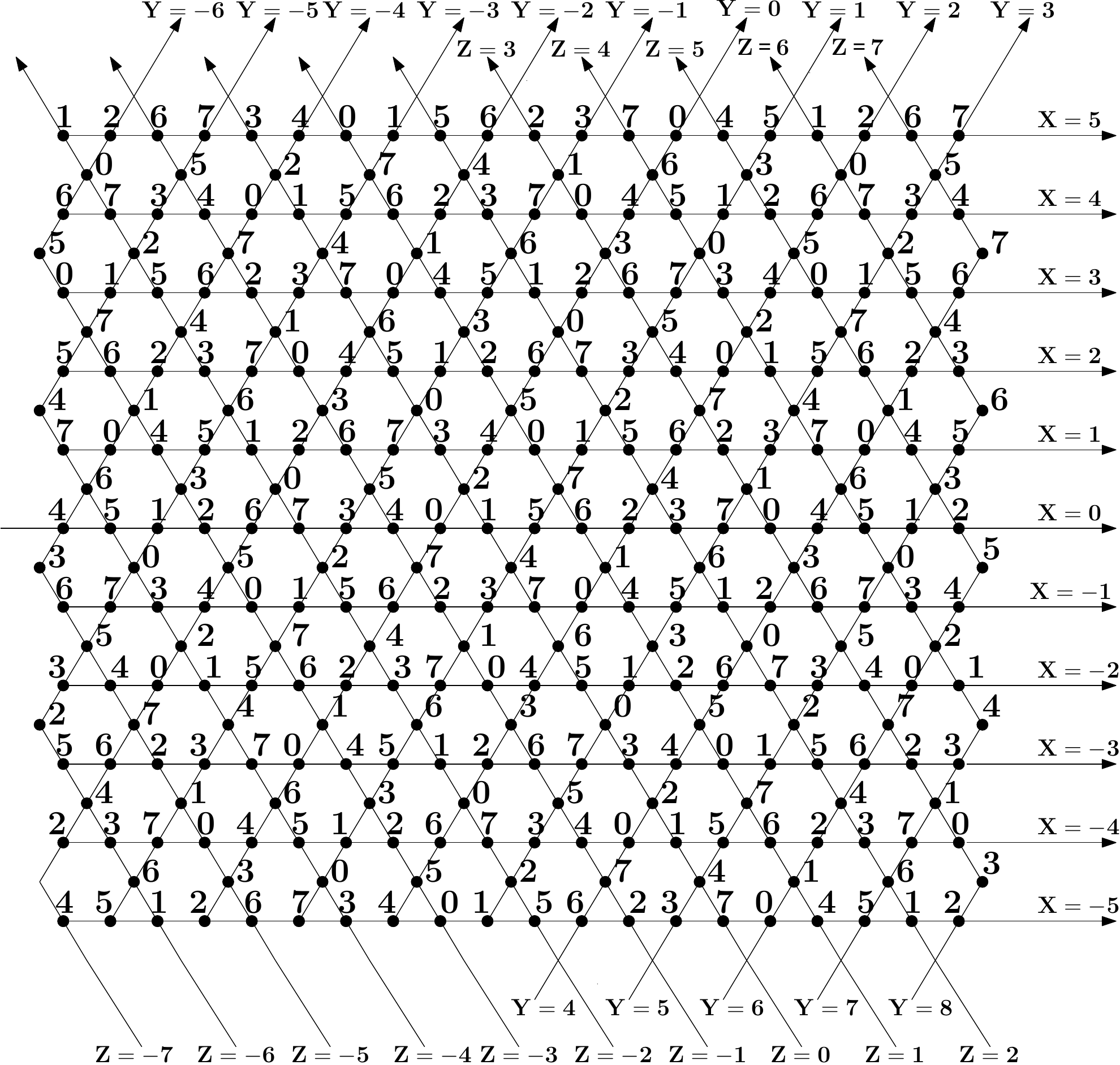}
\caption{Sub graph $G_{S}$ of $L(T_3)$}
\label{square2}
\end{center}
\end{figure}

\subsection{Hexagonal grid}

Consider $L(T_3)$ and the three co-ordinate axes $\mathbf{X}, \mathbf{Y}$ and $\mathbf{Z}$ as shown in Figure~\ref{square2}. Each vertex is an intersection of two of the three axes. The vertices of $L(T_3)$ can be partitioned into three disjoint sets $U_{xy}, V_{yz}$ and $W_{zx}$ as defined bellow:\\
$U_{xy}=\{u_{xy}: u_{xy}\ is\ an\ intersection\ of\ \mathbf{X=x}\ and\ \mathbf{Y=y}\}, \\
V_{yz}=\{v_{yz}: v_{yz}\ is\ an\ intersection\ of\ \mathbf{Y=y}\ and\ \mathbf{Z=z}\}, \\
W_{zx}=\{w_{zx}: w_{zx}\ is\ an\ intersection\ of\ \mathbf{Z=z}\ and\ \mathbf{X=x}\}. \\$

\begin{theorem}\label{vth2}\footnote{Proof of this theorem is  modified from that of the conference version~\cite{ours}.}
$\lambda'_{1,2}(T_3)= 7$.

\end{theorem}
\begin{proof}
The coloring functions of vertices of $L(T_3)$ are defined as follows.\\
$f(u_{xy})=\left( \left(4 \times \lceil \frac{x}{2}\rceil+ 2 \times \lfloor \frac{x}{2}\rfloor\right)\mod 8 + \left(5 \times y\right)\mod 8 \right)\mod 8, \forall u_{xy}\in U_{xy}.$ \\
 $g(v_{yz})=\left( \left(2+3\times z \right)\mod 8 + \left(  2\times y \right)\mod 8 \right)\mod 8, \forall v_{yz}\in V_{yz}.$ \\
 $h(w_{zx})=\left( \left(1+5\times z \right)\mod 8 + \left(  2\times x \right)\mod 8 \right)\mod 8, \forall w_{zx}\in W_{zx}.$
 
Here we consider that $0\leq (x\mod y) < y$ where $x\in \mathbb{Z}$ and $y\in \mathbb{Z}\setminus \{0\}$.
 
The colors of the vertices of a finite subgraph $G_S$ of $L(T_3)$ are shown in Figure~\ref{square2}. It can be verified that colors of every pair of vertices satisfy all the $L(1,2)$-vertex labeling constraints. It is also evident that the colors obey a regular modulo pattern which can be extended up to infinity and there will be no color conflict between any pair of vertices of $L(T_3)$ if the assigned colors satisfy the coloring functions. The minimum and maximum color used here are $0$ and $7$ respectively. Hence $\lambda_{1,2}(L(T_3)) \leq 7$. It has been shown in \cite{lin3} that $\lambda_{1,2}(L(T_3))\geq 7$. Hence $\lambda'_{1,2}(T_3)=\lambda_{1,2}(L(T_3))= 7$. 
 \end{proof}
 
\subsection{Square grid}

Let us consider the induced subgraph $G$ of $L(T_4)$ as shown in Fig~\ref{sub_12_ver} where all vertices are at mutual distance at most three. 
Let $S_1=\{a,b\}$, $S_2=\{k,l\}$, $S_3=\{c,g\}$, $S_4=\{f,j\}$ and $S_5=\{d,\ e,\ h,\ i\}$.

\begin{definition}\label{vdef1}
The set of vertices in $S_5$ are termed as \textbf{central vertices} in $G$.
\end{definition}

\begin{definition}\label{vdef2}
The set of vertices in $S_1 \cup S_2 \cup S_3 \cup S_4$ are termed as \textbf{peripheral vertices} in $G$.
\end{definition}

Now we have the following observations in $G$. Here the color of vertex $a$ is denoted by $\mathbf{f}(a)$. 

\begin{center}
\begin{figure}[h!]
\centerline{\includegraphics[width=4cm]{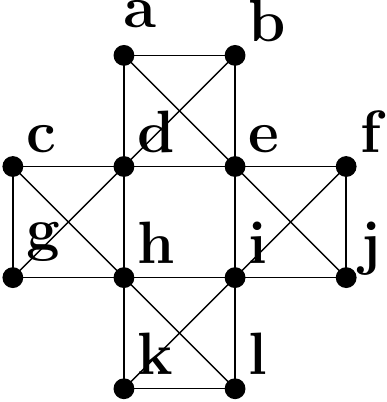}}
\caption{Sub graph $G$ of $L(T_4)$ }
\label{sub_12_ver} 
\end{figure} 
\end{center}

\begin{observation}\label{vobs1}:
If colors of vertices of $G$ are all distinct then $\lambda_{1,2}(G)\geq 11$.
\end{observation}
\begin{proof}
As $G$ has $12$ vertices, if all of them get distinct colors then  $\lambda_{1,2}(G)\geq 11$.
\end{proof}

\begin{observation}\label{vobs2}:
No color can be used thrice in $G$. Colors used at the central vertices in $S_5$ can not be reused in $G$. Colors used at the peripheral vertices in $S_1$ can be reused only at the peripheral vertices in $S_2$. Similarly, colors used at the peripheral vertices in $S_3$ can be reused only at the peripheral vertices in $S_4$.
\end{observation}

\begin{proof}
No three vertices are mutually distant three apart. Hence no color can be used thrice in $G$. For any central vertex in $S_5$ there does not exist any vertex in $G$ which is distance three apart from it. So colors used in the central vertices in $S_5$ can not be reused in $G$. For all peripheral vertices in $S_1 \cup S_2$, $d(x,y)=3$ only when $x \in S_1$ and $y \in S_2$. Hence color used at peripheral vertex in $S_1$ can only be reused in $S_2$. Similarly, color used at peripheral vertex in $S_3$ can only be reused in $S_4$.
\end{proof}

\begin{observation}\label{vobs2_1}:
If $\mathbf{f}(x)=\mathbf{f}(y)=\mathbf{c}$ where $x \in S_1$ and $y \in S_2$ then either $\mathbf{c} \pm 1$ is to be used in $(S_1 \cup S_2) \setminus \{x,y\}$ or it should remain unused in $G$. Similarly, if $\mathbf{f}(x)=\mathbf{f}(y)=\mathbf{c}$ where $x \in S_3$ and $y \in S_4$ then either $\mathbf{c} \pm 1$ is to be used in $(S_3 \cup S_4) \setminus \{x,y\}$ or it should remain unused in $G$.
\end{observation}

\begin{proof}
Note that for all vertices $z \in V(G) \setminus (S_1 \cup S_2)$, either $d(z,x)=2$ or $d(z,y)=2$, where $x \in S_1$ and $y \in S_2$. Hence $\mathbf{c} \pm 1$ can not be used in  $V(G) \setminus (S_1 \cup S_2)$. So $\mathbf{c} \pm 1$ can only be used in $(S_1 \cup S_2) \setminus \{x,y\}$ or it should remain unused in $G$.  Similarly, if $\mathbf{f}(x)=\mathbf{f}(y)=\mathbf{c}$, where $x \in S_3$ and $y \in S_4$, then $\mathbf{c} \pm 1$ can only be used in $(S_3 \cup S_4) \setminus \{x,y\}$ or it should remain unused in $G$.
\end{proof}

\begin{observation}\label{vobs2_2}:
Let $\mathbf{f}(x)=\mathbf{f}(y)=\mathbf{c}$ where $x \in S_1$ and $y \in S_2$. If $|\mathbf{f}(x)-\mathbf{f}(x')| \geq 2$, where $x' \in S_1 \setminus \{x\}$, then one of $\mathbf{c} \pm 1$ must remain unused in $G$. Similarly if $|\mathbf{f}(y)-\mathbf{f}(y')| \geq 2$, where $y' \in S_2 \setminus \{y\}$, then one of $\mathbf{c} \pm 1$ must remain unused in $G$. Similar facts hold when $x \in S_3$, $x' \in S_3 \setminus \{x\}$, $y \in S_4$ and $y' \in S_4 \setminus \{y\}$.
\end{observation}

\begin{proof}
Since $|\mathbf{f}(x)-\mathbf{f}(x')| \geq 2$, $\mathbf{f}(x')\neq \mathbf{c} \pm 1$. Hence from observation \ref{vobs2_1}, one of $\mathbf{c} \pm 1$ must remain unused in $G$.  
\end{proof}

If no color is reused in $G$, then $\lambda_{1,2}(G)\geq 11$ from observation~\ref{vobs1}. To make $\lambda_{1,2}(G) < 11$, at least one color must be reused in $G$. From observation~\ref{vobs2}, there are at most $4$ distinct pairs of peripheral vertices in $G$ where a pair can have the same color. Now consider the subgraph $G_1$ of $L(T_4)$ as shown in Fig.~\ref{sub_12_ver_extnd}.a. Note that $G_1$ consists of $5$ subgraphs $G'$, $G'_1$, $G'_2$, $G'_3$ and $G'_4$ which all are isomorphic to $G$ having central vertices $\{ d,h,i,e\}$, $\{t_1,c,d,a \}$, $\{ b,e,f,t_2\}$, $\{ i,l,t_3,j\}$ and $\{g,t_4,k,h \}$ respectively. Based on the span requirements of coloring $G_1$, we derive the following theorem.


\begin{center}
\begin{figure}[h!]
\centerline{\includegraphics[width=12cm]{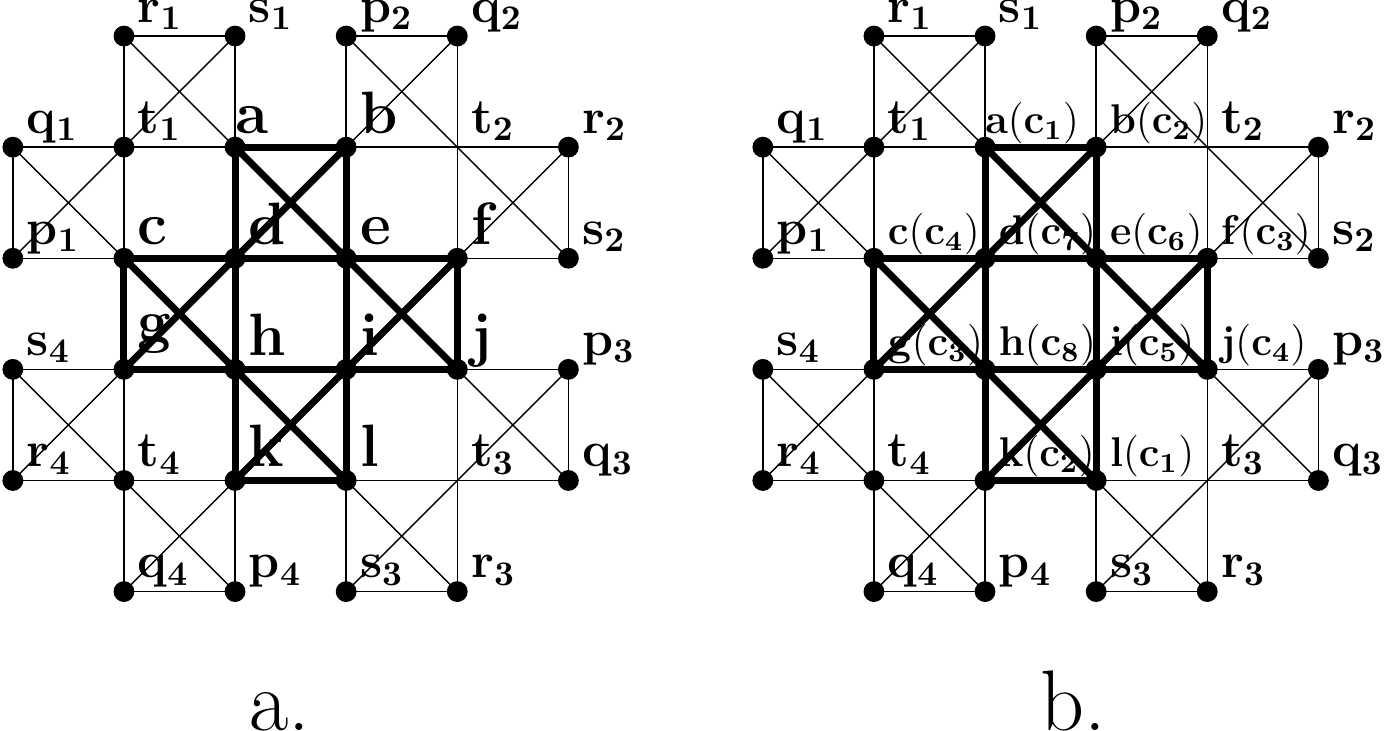}}
\caption{Different coloring of subgraph $G_1$. }
\label{sub_12_ver_extnd} 
\end{figure} 
\end{center}

\begin{theorem}
\label{vth1}$\lambda_{1,2}(L(T_4))= \lambda_{1,2}(G_1)= 11$.

\end{theorem} 
\begin{proof}

Due to He and Lin~\cite{lin3}, it is already known that $\lambda_{1,2}(L(T_4))\leq 11$. Thus in this section our main goal will be to improve the lower bound and show $\lambda_{1,2}(L(T_4))\geq 11$. We do some case analysis depending on how the peripheral vertices are colored. 
 \newline
\textbf{Case 1: When at most one pair of peripheral vertices use the same color in any sub graph of $L(T_4)$ isomorphic to $G$.}\\ If no color is reused in $G'$, then $\lambda_{1,2}(G')\geq 11$ from observation~\ref{vobs1}. We now consider the case when exactly one pair reuse a color in $G'$. Without loss of generality, consider $\mathbf{f}(a)=\mathbf{f}(l)=c_1$. From Observation~\ref{vobs2_1}, $c_1 \pm 1$ can only be put in $\{b,k\}$. Let $\mathbf{f}(k)=c_1-1$ and $\mathbf{f}(b)=c_1+1$. We assume that $c_1-1$ is the minimum color. Let us consider $\mathbf{f}(d)=c_1+n$ where $n \in \mathbb{N}$ and $n \geq 2$. From observation~\ref{vobs2_2}, $x \in \{c_1,c_1+n\}$ can be reused in $G'_2$ only if one of $x \pm 1$ remains unused in $G'_2$. In either case, $\lambda_{1,2}(G'_2)\geq 11$. So $x$ can not be reused in $G'_2$. Since $\mathbf{f}(a)=\mathbf{f}(l)=c_1$, $c_1-1$ can only be put in $\{r_2,s_2\}$ as vertex $b$ is already colored and for all other vertices $z \in V(G'_2)\setminus \{r_2,s_2\}$, either $d(z,a)=2$ or $d(z,l)=2$. Without loss of generality, let $\mathbf{f}(r_2)=c_1-1$. In that case, $c_1+n \pm 1$ can only be put in $\{e,s_2\}$. Without loss of generality, let $\mathbf{f}(e)=c_1+n-1$ and $\mathbf{f}(s_2)=c_1+n+1$. Since $\mathbf{f}(a)=\mathbf{f}(l)=c_1$, $\mathbf{f}(i)\neq c_1 \pm 1$ and hence $\vert \mathbf{f}(l)-\mathbf{f}(i) \vert \geq  2$. Now if $\vert \mathbf{f}(d)- \mathbf{f}(c) \vert \geq 2$, then from observation ~\ref{vobs2_2}, one of $\mathbf{f}(c)\pm 1$, $\mathbf{f}(d)\pm 1$ and $\mathbf{f}(i)\pm 1$ remains unused in $G'_4$ if $\mathbf{f}(c)$ or $\mathbf{f}(d)$ or $\mathbf{f}(i)$ is reused in $G'_4$ respectively. In either case, this implies $\lambda_{1,2}(G'_4)\geq 11$. So $\vert \mathbf{f}(d)- \mathbf{f}(c) \vert = 1$ and $\mathbf{f}(c)=c_1+n+1$. There are $5$ more vertices $\{g,h,i,j,f\}$ in $G'$ which are to be colored with $5$ distinct colors. Hence at least color $c_1+n+6$ must be used. Observe that if $\mathbf{f}(f)=c_1+n+2$ then $\vert \mathbf{f}(e)-\mathbf{f}(f)\vert = 3$ and $\vert \mathbf{f}(k)-\mathbf{f}(h)\vert \geq 3$ implying  $\lambda_{1,2}(G'_3)\geq 11$ from observation~\ref{vobs2_2}. As $d(s_2,i)=d(s_2,j)=2$ and $\mathbf{f}(s_2)=c_1+n+1$, we get $\mathbf{f}(i) \neq c_1+n+2$ and $\mathbf{f}(j) \neq c_1+n+2$. Therefore, either $\mathbf{f}(g)=c_1+n+2$ or $\mathbf{f}(h)=c_1+n+2$. So, $\mathbf{f}(p_4) \neq c_1+n+1$ and $\mathbf{f}(q_4) \neq c_1+n+1$. In that case, $\mathbf{f}(p_4)$ and $\mathbf{f}(q_4)$ must be in $\{c_1+n,c_1+n-1\}$ if color $c_1+n$ is to be reused in $G'_4$, otherwise, $\lambda_{1,2}(G_1)\geq 11$. As $c_1$ can not be reused in $G'_4$, either $\mathbf{f}(r_4)=c_1+1$ or $\mathbf{f}(s_4)=c_1+1$. Let $\mathbf{f}(r_4)=c_1+1$. When $n=2$, $c_1+n-1=c_1+1$ and when $n=3$, $c_1+n-1=c_1+2$. As $d(p_4,l)=d(p_4,r_4)=d(q_4,l)=d(q_4,r_4)=2$, $\mathbf{f}(p_4), \mathbf{f}(q_4) \notin \{c_1+1,c_1+2\}$. So, $n \geq 4$ and hence $c_1+n+6 \geq c_1+10$. So at least $12$ color are required in $G_1$ including $c_1-1$ and $c_1+10$. Hence $\lambda_{1,2}(G_2)\geq 11$.\\

\textbf{Case 2: There exists at least one subgraph of $L(T_4)$ isomorphic to $G$ where two pairs of peripheral vertices use a color each.} \\ There are two different ways of reusing two colors in $G'$.\\
{\bf Case 2.1}: First consider the case when $\mathbf{f}(a)=\mathbf{f}(l)=c_1$ and $\mathbf{f}(c)=\mathbf{f}(j)=c_2$. From observation \ref{vobs2_1}, $c_1 \pm 1$ and $c_2\pm 1$ must be used in $\{b,k\}$ and $\{g,f\}$ respectively. From observation~\ref{vobs2}, $c_1$ can only be reused in $\{r_2,s_2\}$ in $G'_2$. But $\mathbf{f}(r_2)\neq c_1$ and $\mathbf{f}(s_2)\neq c_1$ as $\vert \mathbf{f}(b)-c_1 \vert=1$ and $d(b,r_2)=d(b,s_2)=2$. Again, from observation~\ref{vobs2}, $c_2$ can only be reused in $\{p_2,q_2\}$. But $\mathbf{f}(p_2)\neq c_2$ and $\mathbf{f}(q_2)\neq c_2$ as $\vert \mathbf{f}(f)-c_2 \vert=1$ and $d(f,p_2)=d(f,q_2)=2$. From observation \ref{vobs2_1}, if $\mathbf{f}(i)$ is to be reused in $G'_2$, then $\vert \mathbf{f}(i)-c_2 \vert=1$. But $\mathbf{f}(i) \neq c_2 \pm 1$ as $d(c,i)=2$ and $\mathbf{f}(c)=c_2$. If $\mathbf{f}(d)$ is to be reused in $G'_2$, then $\vert \mathbf{f}(d)-c_1 \vert=1$. But $\mathbf{f}(d) \neq c_1 \pm 1$ as $d(d,l)=2$ and $\mathbf{f}(l)=c_1$. Therefore, no color can be reused in $G'_2$ and hence $\lambda_{1,2}(G_1) \geq 11$.

{\bf Case 2.2}: Consider the case when $\mathbf{f}(a)=\mathbf{f}(l)=c_1$ and $\mathbf{f}(b)=\mathbf{f}(k)=c_2$. Without loss of generality, assume   $c_2 > c_1$. From observation \ref{vobs2_1}, $c_1 \pm 1$ and $c_2 \pm 1$ must be used in $\{b,k\}$ and $\{a,l\}$ respectively. Even if we set $c_2=c_1+1$,  at least one of $c_1-1$ and $c_2+1$ must remain unused in $G'$. So the $8$ vertices in $V(G') \setminus (\{a,l\}\cup \{b,k\})$ must get $8$ distinct colors other than $c_1$ and $c_2$. So, $\lambda_{1,2}(G') \geq 10$. Note that  $\lambda_{1,2}(G') = 10$ only if  $c_2 = c_1 + 1$, $c_1$ is minimum color ($c_1-1$ does not exists) or $c_2$ is maximum color ($c_2+1$ does not exists). If both $c_1$ and $c_2$ are non-extreme color, then  $\lambda_{1,2}(G') \geq 11$ and we are done. So, we consider $c_1=0$, $c_2=c_1+1=1$ and $c_2+1=2$ as unused in $G'$. 
In that case, $\mathbf{f}(d)=x \geq 3$ and hence $\vert \mathbf{f}(d)-\mathbf{f}(a)\vert \geq 3$. From observation \ref{vobs2_2}, if $x$ is reused in $G'_2$, then one of $x\pm 1$ can not be used in $G'_2$. If only $x$ is reused in $G'_2$, then $\lambda_{1,2}(G'_2) \geq 11$. If $x$ and one of $\{\mathbf{f}(i),\mathbf{f}(j)\}$ are reused in $G'_2$, then from Case 2.1 above, $\lambda_{1,2}(G_1) \geq 11$. If $x$ and both of $\{\mathbf{f}(i),\mathbf{f}(j)\}$ are reused in $G'_2$, from Case 3 below, we will see that $\lambda_{1,2}(G_1) \geq 11$. So, to keep $\lambda_{1,2}(G_1) < 11$, $x$ should not be reused in $G'_2$. In that case, $x-1$ must be used at one of $\{c,g,h,e \}$ in $G'$. Now arguing similarly as stated in case $\mathbf{1}$, we can conclude that $x+7$ must be used in $G'_1$ or $G'_2$. If $x=3$, then $x-1=2$ must be used in $G'$ which is a contradiction, as $2$ must remain unused in $G'$. Hence $x \geq 4$ implying $x+7=11$. Hence $\lambda_{1,2}(G_1) \geq 11$. 

\textbf{Case 3: The exists at least one sub graph of $L(T_4)$ isomorphic to $G$ where three pairs of peripheral vertices use a color each.}\\ 
Without loss of generality, let us consider $\mathbf{f}(a)=\mathbf{f}(l)=c_1$, $\mathbf{f}(b)=\mathbf{f}(k)=c_2$ and $\mathbf{f}(c)=\mathbf{f}(j)=c_3$. From observation \ref{vobs2_1}, $c_1 \pm 1$ and $c_2 \pm 1$ must be used in $\{b,k\}$ and $\{a,l\}$ respectively. It can be observed that $\lambda_{1,2}(G') = 9$ only if  $\vert c_1-c_2 \vert =1$, $\vert c_3-\mathbf{f}(g) \vert =1$, $\vert c_3-\mathbf{f}(f)\vert=1$ and any one of $\{c_1,c_2\}$ is one extreme color. Without loss of generality consider $\mathbf{f}(g)=c_3+1$, $\mathbf{f}(f)=c_3-1$, $c_1$ is minimum color and $c_2=c_1+1$. From observation~\ref{vobs2}, $c_3$ can only be reused in $\{p_2,q_2\}$. But $\mathbf{f}(p_2)\neq c_3$ and $\mathbf{f}(q_2)\neq c_3$ as $\mathbf{f}(f)=c_3-1$ and $d(f,p_2)=d(f,q_2)=2$. From observation \ref{vobs2_1}, if $\mathbf{f}(i)$ is to be reused in $G'_2$, then $\vert \mathbf{f}(i)-c_3 \vert=1$. But $\mathbf{f}(i) \neq c_3 \pm 1$ as $d(c,i)=2$ and $\mathbf{f}(c)=c_3$. From observation~\ref{vobs2}, $c_1$ can only be reused in $\{r_2,s_2\}$. But  $\mathbf{f}(r_2)\neq c_1$ and $\mathbf{f}(s_2)\neq c_1$ as $\mathbf{f}(b)=c_2=c_1+1$ and $d(b,r_2)=d(b,s_2)=2$. Now arguing similarly as stated in case $\mathbf{2.2}$ above, we can conclude that $c_2+1$ must remain  unused in $G'$. So, $(c_1 - \mathbf{f}(d)) \geq 3$. Now from observation~\ref{vobs2_2}, if $\mathbf{f}(d)$ is reused in $G'_2$ then any one of $\mathbf{f}(d) \pm 1$ must remain unused in $G'_2$.  Thus in $G'_2$, only $\mathbf{f}(d)$ can be reused by keeping one of $\mathbf{f}(d)\pm 1$ as unused. Hence $\lambda_{1,2}(G_1) \geq 11$.  If we consider $\lambda_{1,2}(G') = 10$, the same result can be obtained by considering the corresponding $G'_i$, $1 \leq i \leq 4$.

\textbf{Case 4: The exists at least one subgraph of $L(T_4)$ isomorphic to $G$ where all four pairs of peripheral vertices use a color each.}\\ Let us consider $\mathbf{f}(a)=\mathbf{f}(l)=c_1$, $\mathbf{f}(b)=\mathbf{f}(k)=c_2$, $\mathbf{f}(g)=\mathbf{f}(f)=c_3$ and $\mathbf{f}(c)=\mathbf{f}(j)=c_4$. From observation \ref{vobs2_1}, $c_1\pm 1$, $c_2 \pm 1$, $c_3\pm 1$ and $c_4\pm 1$ must be used in $\{b,k\}$, $\{a,l\}$, $\{c,j\}$ and $\{g,f\}$ respectively. It can be observed that $\lambda_{1,2}(G') = 9$ only if  $\vert c_1-c_2 \vert =1$, $\vert c_3 -c_4\vert =1$, one of $\{c_1,c_2\}$ is an extreme color and one of $\{c_3,c_4\}$ is the other extreme color. Without loss of generality, consider $c_1=0$, $c_4=9$, $c_2=c_1+1=1$ and $c_3=c_4-1=8$. So $c_2+1=2$ and $c_3-1=7$ are two distinct unused colors. Without loss of generality, consider $c_8=c_2+2$, $c_5=c_8+1$, $c_6=c_5+1$ and $c_7=c_6+1$. Since $\vert c_3 -c_4\vert =1$ and $d(g,p_4)=d(g,q_4)=2$, we get $\mathbf{f}(p_4) \neq c_4$ and $\mathbf{f}(q_4)\neq c_4$. Similarly, $\mathbf{f}(r_4) \neq c_1$ and $\mathbf{f}(s_4)\neq c_1$. From observation \ref{vobs2}, $c_5$ can only be reused at $\{s_4,r_4\}$ in $G'_4$ but $\mathbf{f}(s_4) \neq c_5$ and $\mathbf{f}(r_4)\neq c_5$ as $d(h,s_4)=d(h,r_4)=2$ and $\mathbf{f}(h)=c_8=c_5-1$. Therefore, only $c_7$ can be reused in $\{p_4,q_4\}$. From observation \ref{vobs2_2}, one of $c_7 \pm 1$ must remain unused in $G'_4$ as $(c_4 - c_7) = 3$. Hence $\lambda_{1,2}(G_1) \geq 11$. For other assignment of central vertices and for the case when $\lambda_{1,2}(G') = 10$, we can obtain the same result by considering the corresponding $G'_i$, $1 \leq i \leq 4$. 
\end{proof}

\subsection{Triangular Grid}
For any vertex $u$, the set of vertices which are adjacent to $u$ is called $N(u)$. Let us define $N(S)=\{\cup_{u\in S} N(u) : u \in S\}$. Let $v$ be any vertex in $T_6$. Consider the subgraph $G_v(V,E)$ of $T_6$ centering $v$ as shown in Figure \ref{Fig1}, where $V = N(v) \cup N(N(v))$ and $E$ is set of all the edges which are incident to $u$ where $u \in N(v)$. Observe that in $G_v$, for any two edges $e_1$ and $e_2$, $d(e_1,e_2)\leq 3$.  Now we define the following three sets of edges $S_1$, $S_2$ and $S_3$:

\begin{figure}[h!]
\centering  
\includegraphics[width=6cm]{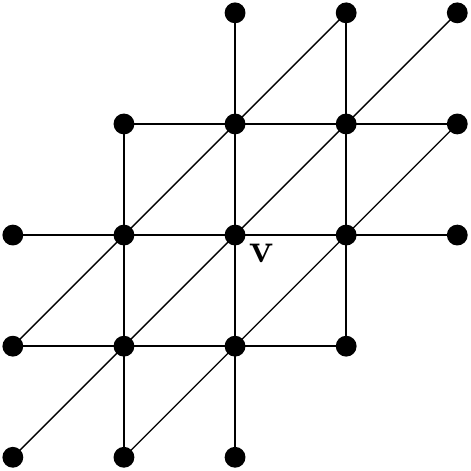}
\caption{ A subgraph $G_v$ of $T_6$}
\label{Fig1}
\end{figure}

\textit{$S_1$}: Edges of $G_v$ incident to $v$. 

\textit{$S_2$}: Edges of $G_v$ whose both end points incident to $e_1$ and $e_2$ where $e_1,\ e_2 \in S_1$.

\textit{$S_3$}: $E \setminus (S_1\cup S_2)$.

 Consider the $6$-cycle, $H_v$ formed with the edges of $S_2$ in $G_v$.  We say $e$ and $e_1$ as a pair of \textit{opposite edges} in $H_v$ iff $d(e,e_1)=3$. This implies that the same color can be used at a pair of opposite edges in $L(1,2)$-edge labeling. 
 An edge $e$ \textit{covers} the set of edges $E'$ if for every $e' \in E'$, $d(e,e')\leq 2$.
 This implies that a color used at $e$ can not be used at any edge $e' \in E'$ in $L(1,2)$-edge labeling. Now we have the following lemmas.

\begin{lemma}\label{lemma1}
If $c$ is a color used to color an edge $e$ in $S_1$, then $c$ can not be used in $E \setminus e$.
\end{lemma}
\begin{proof}
Since $e$ is incident to $v$,  for any other edge $e_1 \in E$, $d(e,e_1) \leq 2$. Hence $f'(e_1) \neq c$ for $L(1,2)$-edge labeling, where $f'(e_1)$ denotes the color of $e_1$. \end{proof}

\begin{lemma}\label{lemma2}
If $c$ is a color used to color an edge in $S_1$, then $c + 1$ and $c-1$ both can be used at most once in $G_v$.
\end{lemma}
\begin{proof}
Let $e$ be an edge in $S_1$ such that $f'(e)=c$. Since $e$ is incident to $v$, for any other edge $e_1 \in E$, $d(e,e_1) \leq 2$. 
Let $S_e = \{e_1: d(e,e_1)=1\}$. For $L(1,2)$-edge labeling, $c + 1$ can only be used in an edge $e_1$ in $S_e$. It can be noted that for any two edges $e_1, e_2 \in S_e$, $d(e_1,e_2)\leq 2$. Hence $c + 1$ can be used at most once. Proof for $c-1$ can be done in similar manner. \end{proof}

\begin{lemma}\label{lemma3}
If $c$ is a color used to color an edge $e$ in $S_2$, then $c$ can be used at most one edge in $E \setminus e$ in $G_v$.
\end{lemma}
\begin{proof}
Note that $c$ can not be used at any edge in $S_1$. Here $c$ can be used at the \textit{opposite edge} $e_1$ of $e$ in $S_2$ or at an edge $e_2$ in $S_3$, which is adjacent to $e_1$.
When $c$ is used at $e$ and $e_1$, then $c$ can not be used again in $G_v$ as $e$ and $e_1$ together cover all the edges of $G_v$. When $c$ is used  at $e$ and $e_2$, $c$ can not be used again in $G_v$  as $e$ and $e_2$ together also cover all the edges of $G_v$. \end{proof}

\begin{lemma}\label{lemma4}
If $c$ is a color used to color an edge $e$ in $S_2$, then $c+1$ and $c-1$ both can be used at most twice in $G_v$.
\end{lemma}

\begin{proof}
Suppose $e_1$ be an edge colored with $c + 1$. If $e_1$ is not adjacent to $e$ then $d(e_1, e)=3$. From statement of lemma~\ref{lemma3}, it follows that there does not exist two edges along with $e$ in $G_v$ which are mutually distance $3$ apart, otherwise $c$ would have been used for three times. Hence $c+1$ can be used at most once.

When $e_1$ is adjacent to $e$, $e_2$ can be colored with $c+1$ if $e_2$ is at distance $3$ apart from both $e_1$ and $e$. Again from the statement of lemma~\ref{lemma3}, it follows that there does not exist two edges along with $e$ in $G_v$ which are mutually distance $3$ apart, otherwise $c$ would have been used for three times. So, $c+1$ can be used at most twice, one in one of the edges adjacent to $e$ and other in one of the edges which are at distance $3$ apart from $e$. Proof for $c-1$ can be done in similar manner.
\end{proof}

\begin{lemma}\label{lemma5}
If $c$ is a color used to color an edge $e$ in $S_3$, then $c$ can be used at most twice in $E \setminus e$.
\end{lemma}
\begin{proof}
It follows from Figure~\ref{Fig1} that exactly one end point of $e$ is incident to a vertex in $H_v$. Note that for any walk through $H_v$, every third vertex is distance $2$ apart. So edges incident to those vertices are distance $3$ apart. Since the order of $H_v$ is $6$, there can be at most $6/2=3$ vertices which are mutually distance $2$ apart. Hence $c$ can be used thrice.  
\end{proof}

\begin{lemma}\label{lemma6}
If $c$ is a color used to color an edge $e$ in $S_3$, then $c + 1$ and $c-1$ both can be used at most thrice in $G_v$.
\end{lemma}
\begin{proof}
We know that $c+1$ can be used at an edge adjacent to $e$. From lemma~\ref{lemma5} it is clear that $c$ can be used at most thrice. So, $c+1$  can also be used at most thrice, where each such edge is adjacent to one of the three edges colored with $c$. It can be proved similarly for $c-1$.
\end{proof}

\begin{lemma}\label{lemma7}
\begin{enumerate}
\item[i.] To color the edges of $S_1$, at least $6$ colors are required.

\item[ii.] To color the edges of $S_2$, at least $3$ colors are required.

\item[iii.] To color the edges of $S_3$, at least $6$ colors are required.
\end{enumerate}
\end{lemma}
\begin{proof}
i. From lemma~\ref{lemma1}, every edge of $S_1$ has an unique color. As there are $6$ edges in $S_1$, $6$ distinct colors are required here.

ii. In $S_2$, there are $3$ pairs of {\it opposite edges}. Each pair of opposite edges requires at least one unique color. So at least $3$ colors are required.

iii. A color can be used thrice in $S_3$ by lemma~\ref{lemma5}. In $S_3$, there are $18$ edges. So, at least $6$ colors are required.

\end{proof}
\begin{theorem}\label{th3}
For any optimal labeling of $G_v$, $6$ consecutive colors including either the minimum color or the maximum color must be used in $S_1$.
\end{theorem}
\begin{proof}
It is clear from lemma~\ref{lemma7}.i that $S_1$ needs at least $6$ colors to color its edges. From lemma~\ref{lemma2}, note that if $c$ is a color used in an edge of $S_1$ then both $c+1$ and $c-1$ can be used at most once in $G_v$. Whereas a color can be used twice in $S_2$ and thrice in $S_3$. Thus our aim should be to minimize the number of colors which can be used only once in $G_v$. This implies that consecutive colors should be used in $S_1$ for optimal coloring. If the minimum color ($min$) or the maximum color ($max$) is used in $S_1$ then further benefit can be achieve as $min-1$ or $max+1$ does not exist. Therefore, optimal span can be achieved only when the colors of $S_1$ are consecutive including either $min$ or $max$. 
\end{proof}

\begin{lemma}\label{lemma9a}
If three consecutive colors $c$, $c+1$, $c+2$ are used thrice each in $S_3$ then neither $c-1$ nor $c+3$ can be used in $S_3$.
\end{lemma}
\begin{proof}
Observe that there are exactly $2$ sets of three alternating vertices in $H_v$ where a color can be used thrice at edges incident to any set of alternating vertices. If $c-1$ would have been used in $S_3$ then either it was used at an edge adjacent to the edges colored with $c$ or at an edge distance $3$ apart from the edge colored with $c$.  Now observe that $c$ and $c-1$ are used at two edges of $S_3$ which form a triangle with one edge of $S_2$. Suppose $c$, $c-1$ be the colors used at those two edges  $e,e_1 \in S_3$ respectively, where $e$ is incident to $u$ and $e_1$ is incident to $w$ where $uw \in S_2$. Note that $c$ is used thrice in $S_3$. Then $c$ must be reused at an edge incident to $x$, and  $xw\in S_2$. So $c$ and $c-1$ are used at two edges at distance $2$ apart, which violets the condition of $L(1,2)$-edge labeling. Hence $c-1$ can not be used in $G_v$. Similarly it can be shown that $c+3$ can also not be used in $G_v$. This implies that no $4$ consecutive colors can be used thrice each in $G_v$.
\end{proof}

\begin{theorem}\label{th1}\footnote{Theorem is  modified from that of the conference version~\cite{ours}.}
$\lambda'_{1,2}(G_v) \geq 16$.
\end{theorem}

\begin{proof}
By Theorem~\ref{th3}, $6$ consecutive colors must be used to color the edges of $S_1$. Recall that, we assume the minimum color is used at $S_1$. Let $c_1$ be the maximum color used in $S_1$ and $c$ be the minimum color used in $S_2$. Now we consider the coloring of the edges in $S_2 \cup S_3$. If the edges of $S_2$ have all consecutive colors $c,c+1, \cdots, c+5$ then $6$ colors are needed for $S_2$. From Lemma~\ref{lemma3}, any color $c'$, used in $S_2$ can be reused at most once more in $G_v$ unless $c'=c_1+1$ (Lemma \ref{lemma2}). Thus we can color at most $6$ edges in $S_3$ using those colors. From Lemma~\ref{lemma4}, the color $c+6$ can be used at most twice in $G_v$. Note that color $c_1+1$ can be used at most once in $G_v$. So far, at most $9$ edges in $S_3$ are colored. So at least $9$ edges are left to be colored in $S_3$. So at least $3$ more colors are needed for $S_3$, as any color can be used at most thrice in $S_3$. However, from Lemma~\ref{lemma9a}, in that case, all of $c+7$, $c+8$ and $c+9$ can not be used thrice each in $S_3$. Thus at least $c+10$ is needed for $G_v$. Since $c-6$ is used in $S_1$, we get $\lambda'_{1,2}(G_v)\geq (c+10) - (c-6)=16$. Now, consider the case when $4$ colors are used in $S_2$. Observe that then in any possible coloring, at least $3$ colors can be used at most twice in $S_3$. Therefore, at most $9$ edges of $S_3$ can be colored. Hence here too, we can argue that $c+10$ must be used in $G_v$. In a similar manner, we can argue that  when five colors are used in $S_2$, the color $c+10$ must be used in $G_v$. Hence in all the cases discussed above, we get $\lambda'_{1,2}(G_v)\geq 16$. 

 Now consider the case when only three colors say $c$, $c'$, $c''$ are used in $S_2$. Without loss of generality assume $c'-c\geq 2$ and $c''-c'\geq 2$. 
 
First consider the cases assuming $c_1+1=c-1$.
Observe that if $c+1\neq c'-1$ and $c'+1\neq c''-1$, then $c\pm 1,c'\pm 1,c''\pm 1$ can color at most $(1+5\times2)=11$ edges of $S_3$. So at least $7$ edges are left to be colored in $S_3$ requiring at least $3$ more colors. So at least $9$ colors are required for $S_3$ and hence $\lambda'_{1,2}(G_v)\geq 17$. If $c''$ is maximum color, then $c''+1$ can not be used and in this case too, at least $17$ colors are required for $G_v$ and hence $\lambda'_{1,2}(G_v)\geq 16$. Now we consider the case when $c+1= c'-1$ but $c'+1 \neq c''-1$. Here $c+1, c'+1, c''\pm 1$ can color at most $8$ edges of $S_3$. Thus using $c-1$ also we are able to color at most $9$ edges in $S_3$. Therefore at least $9$ edges are left to be colored in $S_3$ requiring at least $3$ more colors. So at least $8$ colors are required for $S_3$ and hence $\lambda'_{1,2}(G_v)\geq 16$. Similarly we can argue that $\lambda'_{1,2}(G_v)\geq 16$ when $c+1\neq c'-1$ $c'+1 
=c''-1$. If $c+1= c'-1$ and $c'+1= c''-1$, then $c$, $c+2$ and $c+4$ are used in $S_2$. Here $c\pm 1,c+3,c+5$ can color at most $7$ edges of $S_3$. So at least $11$ edges are left to be colored in $S_3$ requiring at least $4$ more colors. So at least $8$ colors are required for $S_3$ and hence $\lambda'_{1,2}(G_v)\geq 16$. Similarly, one can verify that when $c_1+1\neq c-1$ the bound remains the same for all the cases.
Thus considering all cases, $\lambda'_{1,2}(G_v)\geq 16$.
\end{proof}

We assume that the minimum color is used in $S_1$. The maximum color can be used at most thrice in $S_3$ and at most twice in $S_2$. In all cases, there exists a vertex say $v'$ in $H_v$ such that color of any edge incident to $v'$ is neither minimum nor maximum. Now we consider the subgraph $G_{v'}$ of $T_6$ centering $v'$ and isomorphic to $G_v$. Let $min_1$ and $max_1$ be the minimum and maximum colors used in $S'_1$ in $G_{v'}$.
\begin{lemma}\label{lemma9}\footnote{Lemma is  modified from that of the conference version~\cite{ours}.}
If $max_1-min_1 \geq 7$, i.e., there exists at least two intermediate colors between $min_1$ and $max_1$ which are not used in $S'_1$, then $\lambda'_{1,2}(G_{v'})\geq 18$.
  
\end{lemma}
\begin{proof}
At least two unused colors $c_1,c_2$ are there in $S'_1$ such that $\forall c \in \{c_1,~c_2\}$, either $c+1$ or $c-1$ is used in $S'_1$. From lemma~\ref{lemma2}, $c_1$, $c_2$, $min_1-1$ and $max_1+1$ can be used at most once each in $G_{v'}$. Let $x'$ be the number of colors required for $G_{v'}$. Let us first consider that $6$ colors are used in $S'_2$. These colors can be reused at most once each in $S'_3$. Consider that all of $c_1,~c_2,~min_1-1$ and $max_1+1$ are also used once each in $S'_3$. So at most $(6+4)=10$ edges of $S'_3$ have been colored so far. So at least $(18-10)=8$ edges are left to be colored in $S'_3$ requiring at least $3$ more colors. So, $x' \geq (6 (for\ S'_1) + 6 (for\ S'_2) + 7 (for\ S'_3)) = 19$ and hence $\lambda'_{1,2}(G_{v'})\geq 18$. Now assume $5$ colors are used in $S'_2$. The only possibility is that $4$ colors must be used once each in $S'_2$ and $1$ color should be used two times in $S'_2$. These $4$ colors can be reused at most once each in $S'_3$. So at  most $(4+4)=8$ edges of $S'_3$ have been colored so far. So at least $4$ more colors are required to color the remaining edges of $S'_3$. So, $x' \geq (6+5+8)=19$ and hence $\lambda'_{1,2}(G_{v'})\geq 18$. If $4$ colors are used in $S'_2$, similarly we can show that $\lambda'_{1,2}(G_{v'})\geq 18$.

 Now consider that $3$ colors  are used two times each in $S'_2$. These $3$ colors can not be used anymore in $S'_3$. So at least $(18-4)=14$ edges are left to be colored in $S'_3$, requiring at least $5$ more colors. At least $4$ out of these $5$ colors must be used three times each in $S'_3$. From Lemma~\ref{lemma9a}, if $3$ consecutive colors are used $3$ times each in $S'_3$, then at least one color can not be used in $S'_3$ resulting $\lambda'_{1,2}(G_{v'})\geq 18$. Otherwise, at most two sets of two consecutive colors can be reused three times each in $S'_3$. In that case, at least three colors $x,~y, ~z$ are there such that for all $w\in\{x,~y, ~z\}$ either $w+1$ or $w-1$ is in one of $S'_1$ or $S'_2$. 
 So each of them can not be used three times each. So these $5$ colors can color at most $(3\times 2+ 2\times 3) = 12$ edges in $S'_3$, requiring at least $1$ more color for $S'_3$. So at least $(4+6)=10$ colors are required for $S'_3$ and hence  $\lambda'_{1,2}(G_{v'})\geq 18$. For all other cases  we also get $\lambda'_{1,2}(G_{v'})\geq 18$ using similar argument.
\end{proof}

\begin{theorem}\label{th2}\footnote{Theorem is  modified from that of the conference version~\cite{ours}.}
$\lambda'_{1,2}(T_6)\geq 18$.
\end{theorem}
  \begin{proof}
    Assume that $x$ be a vertex which is not adjacent to edges colored with  maximum and minimum  colors used in $G_x$. Let us consider $G_x$ is not colored and $u$, $w$ be two vertices of $H_x$ in $G_x$. Let us define $S_{x1}$ as the set of edges adjacent to $x$.   We consider the following two cases.

    \textbf{When $w \in N(u)$:} $u$ and $w$ are connected by an edge $e$. Let $\{c_1, \cdots, c_6 \}$ and $\{c'_1, \cdots, c'_6 \}$ be two sequences consisting of consecutive colors are used at the edges incident to $u$ and $w$ respectively. It is possible to assign consecutive colors at those edges when $e$ is colored with either $c_6=c'_1$ or $c_1=c'_6$. Now observe two edges $e'$ and $e'_1$ of $S_{x1}$ are already colored and those are not consecutive. Note that $\vert f'(e')-f'(e'_1)\vert \geq 2$. If $\vert f'(e')-f'(e'_1)\vert = 2$ then $ f'(e')$ and $f'(e'_1)$ is neither minimum nor maximum color used in $u$ and $w$. Then any color of any other edge in $S_{x1}$ is neither consecutive to $f'(e')$ nor $f'(e'_1)$. So $max-min \geq 7$ where $min$ and $max$ be the minimum and maximum colors used in $S_{x1}$. If $\vert f'(e')-f'(e'_1)\vert > 2$, then also $max-min \geq 7$. Hence from lemma~\ref{lemma9}, $\lambda'_{1,2}(G_{x})\geq 18$.

   \textbf{When $w \notin N(u)$:} Note that $x\in \{ N(u)\cap N(w) \} $. Let two sequences $\{c_1, \cdots, c_6 \}$ and $\{c'_1, \cdots, c'_6 \}$ consisting of consecutive colors are used at the edges incident to $u$ and $w$ respectively. Let $uv$ and $wv$ are $e'$ and $e'_1$ respectively. If $f'(e')$ and $f'(e'_1)$ are consecutive then either $ f'(e')=c_6,\ f'(e'_1)=c'_1$ or $\ f'(e')=c_1,\ f'(e'_1)=c'_6$. Now observe that for any other edge $e$ in $S_{x1}$, $\vert f'(e) - f'(e') \vert >2$ implying $max-min \geq 7$ where $min$ and $max$ be the minimum and maximum colors used inc $S_{x1}$. If $f'(e')$ and $f'(e'_1)$ are not consecutive then 
$\vert f'(e')- f'(e'_1) \vert \geq 2 $. If $\vert f'(e')- f'(e'_1) \vert = 2 $ then the intermediate color must be used at an edge $e\in S_{x1}$. There are still $4$ edges remain uncolored. It can be checked that for any coloring of the rest of the graph, there exists a vertex $y \in H_x$ or $ y\in  N(z), z \in H_x$, for which $max-min \geq 7$ where $min$ and $max$ be the minimum and maximum colors used to color the edges incident to $y$ and they are neither maximum nor minimum color used in $G_x$. Hence from lemma~\ref{lemma9}, $\lambda'_{1,2}(G_{x})\geq 18$.
\end{proof}

\subsection{Octagonal Grids}
Graph $G$ shown in figure \ref{T_8} is a subgraph of the infinite octagonal grid $T_8$, consisting of all the edges incident to the vertices $a,b, c, \cdots, p$. Let $H$ be a subgraph of $G$ consisting of all the edges incident to the vertices $g, f, k, j$. Clearly $|E(H)|= 26$, $|E(G)|=86$ and $|E(G \setminus H)|=60$. Observe that the distance between any two edges in $H$ is at most two. Hence in $H$ no color can be used twice. Since $|E(H)|= 26$, we get $\lambda'_{1,2}(T_8) \geq \lambda'_{1,2}(H) \geq 25$. If one color in $\{0, 1, \cdots, 25\}$ is {\it unused} at $H$ then $\lambda'_{1,2}(H) \geq 26$.  Here our goal is to prove $\lambda'_{1,2}(G) \geq 26$. Therefore, we can proceed assuming that no color in $\{0, 1, \cdots, 25\}$ is unused in $H$ and we show that, at least one extra color is needed for proper coloring of $G$. Here, we define a new notion as \textit{fork}. We say that the three edges $e_1=(f,c), e_2=(f,d), e_3=(f,e)$ are forming a fork at vertex $f$ and denoted as $\mathcal{F}_f$. The edges forming the forks $\mathcal{F}_f$, $\mathcal{F}_g$, $\mathcal{F}_j$ and $\mathcal{F}_k$ at vertices $f$, $g$, $j$ and $k$ respectively are all marked by bold edges in the figure. Now, we state the following observations regarding the reusability of the colors used in $H$:
\begin{observation}\label{8reg14}
There exists $14$ edges in $H$, each color used there can be reused at most twice in $G \setminus H$. 
\end{observation}

\begin{proof}
Consider the graph $H'= E(H)\setminus (\mathcal{F}_f\cup \mathcal{F}_g \cup \mathcal{F}_j \cup\mathcal{F}_k)$. Note that $|E(H')|=26-4 \times 3= 14$. Our claim is that each color used in $H'$ can be reused at most twice in $G \setminus H$. One can verify that, 
repetition pattern of the colors of the edges $(b,f)$, $(i,g)$, $(h,j)$, $(o,k)$, $(j,n)$, $(k,e)$ and $(f,l)$ are symmetric to that of the color used in  $(g,c)$. In this sense, edges $(g,j)$, $(j,k)$ and $(k,f)$ are symmetric to $(g,f)$; and edge $(g,k)$ is symmetric to $(j,f)$. So, our goal is to show the following: if color $c'$ is used in any of these three edges then $c'$ can be reused at most twice in $G \setminus H$.
Let us first take the edge, $e'=(g,c)$, and assume, $f'(e')=c'$. The color $c'$ can be used only in some of the edges incident on the vertices $l, m,n, o, p$. Note that $\{l,m,n\}$ and $\{o,p\}$ are two cliques. Hence there can not be more than two edges, where the color $c'$ can be reused without violating the constraints of $L(1,2)$-edge labeling. Now, take the edge, $e'=(g,f)$. Assume that, $f'(e')=c'$. The color $c'$ can be reused in some edges incident on the vertices $m,n,o,p$, where $\{m,n\}$ and $\{o,p\}$ are two cliques. Hence here also, color $c'$ can be reused at most twice in $G \setminus H$. Now we have one more edge to analyse, i.e, $e'= (j,f)$. Let $f'(e')=c'$. There are only two vertices $a$ and $m$, in the graph $G$, whose incident edges can get color $c'$. Hence any color used in $H'$ can be reused at most twice in $G \setminus H$.  
\end{proof}
\begin{center}
\begin{figure}[h!]
\centerline{\includegraphics[width=8cm]{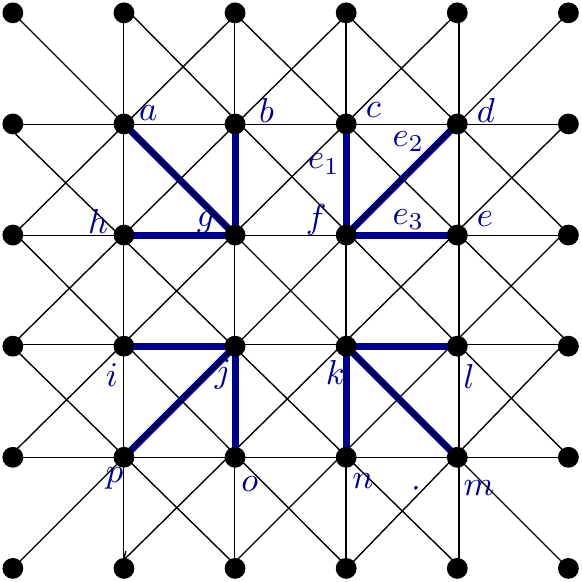}}
\caption{Sub graph $G$ of $T_8$ }
\label{T_8} 
\end{figure} 
\end{center}

\begin{observation}\label{8regfork}
There exists $12$ edges in $H$, each color used there can be reused at most thrice in $G \setminus H$. 
\end{observation}

\begin{proof}
Here we discuss about the edges in $E(H)\setminus E(H')$, i.e., $(\mathcal{F}_f\cup \mathcal{F}_g \cup \mathcal{F}_j \cup\mathcal{F}_k)$. One can observe that, $\mathcal{F}_f$ is symmetric to all other $\mathcal{F}_x, \forall x\in \{g,j,k\}$. So, if we can prove the statement for all edges in $\mathcal{F}_f=\{e_1, e_2, e_3\}$ then we are done. Note that, all the colors used in edges $\{e_1, e_2, e_3\}$ can be reused at some edges incident on the vertices $m,n,o,p,i,h,a$. Note that $\{m,n\}$, $\{o,p,i\}$ and $\{h,a\}$ are three cliques. Hence, there does not exist more than three vertices at a time, such that $c'$ can be used at edges incident on them.    
\end{proof}

As $|E(G \setminus H)|=60$ and $\lambda'_{1,2}(H)=25$, from Observations \ref{8reg14} and \ref{8regfork} it is evident that, to make $\lambda'_{1,2}(G)=25$ as well, we need to reuse at least eight colors ($14 \times 2 + 4\times 2 + 8\times 3=60)$ thrice each in $G \setminus H$. Now, in next two lemmas we discuss about reusability of a color for three times and its constraints.

\begin{lemma}\label{fork3}
If $c'$ be a color used at a fork edge in $H$ and reused thrice in $G \setminus H$ then $c'\pm 1$ must be used at the other two edges of that fork, otherwise $\lambda'_{1,2}(H)\geq 26$.
\end{lemma}

\begin{proof}
From Observations \ref{8reg14} and \ref{8regfork} it is clear that, only the colors used at the four forks $\mathcal{F}_f, \mathcal{F}_g, \mathcal{F}_j,\ \text{and}\ \mathcal{F}_k$ in $H$ can be reused thrice in $G \setminus H$. If we can prove for any one of the four forks then symmetrically it can be shown for others too. So, let us consider the fork $\mathcal{F}_f=\{e_1,e_2,e_3\}$. Without loss of generality assume $f'(e_1)=c'$. Note that  $c'$ can be reused at some edges incident on the vertices $m,n,o,p,i,h,a$. Since $\{m,n\}$, $\{o,p,i\}$ and $\{h,a\}$ are three cliques, $c'$ must be reused at the three edges incident on a vertex from $\{m,n\}$, $\{o,p,i\}$ and $\{h,a\}$ each, respectively. All edges in $H$ which are not incident on $f$ are exactly distance $2$ apart from $e_1$ and hence $c'\pm 1$ can not be used there for any $L(1,2)$-edge labelling. We now consider the possibility of using $c'\pm 1$ at the five edges $(f,b), (f,g), (f,j),(f,k),(f,l)$ which are incident on $f$ but not in $\{e_2,e_3\}$. To reuse $c'$ thrice in $G \setminus H$, $c'$ must be used at an edge incident on a vertex in $\{h,a\}$ but not incident on a vertex in $\{b,g\}$. Both  $(f,b)$ and $(f,g)$ are exactly distance $2$ apart from any such edge. Hence $c'\pm 1$ can not be used at $(f,b)$ and $(f,g)$. Also it is clear that, $c'$ must be reused at an edge incident on a vertex in $\{m,n\}$ but  not incident on a vertex in $\{l,k\}$, which is exactly symmetric to the previous case. So, $c'\pm 1$ can not be used at $(f,l)$ and $(f,k)$ as well. Now consider the case when $c'$ must be used at an edge incident on a vertex in $\{o,p,i\}$ but not incident on $j$. Here also $(f,j)$ is distance $2$ apart from any such edge and hence $c'\pm 1$ can not be used there. Therefore, $c' \pm 1$ have to be used at edges $e_2$ and $e_3$. If we consider $f'(e_2)=c'$ (or, $f'(e_3)=c'$) then also, we can argue in a similar way and show that, other than $e_1$ and $e_3$ (or, $e_1$ and $e_2$), there is no edges in $H$ which can get colors $c' \pm 1$.  However, we must have to use $c'\pm 1$ in $H$ as otherwise $\lambda'_{1,2}(H)\geq 26$.
\end{proof}

Now, if all three colors $c'-1,c', c'+1$ reused thrice each in $G\setminus H$ then from Lemma \ref{fork3} we can conclude that $c'\pm 2$ can not be used in $H$. So, clearly all three of the colors can not be used thrice each. Next lemma says about the reusability of these three colors.

\begin{lemma}\label{fork2}
Only two colors used at a fork in $H$ can be reused thrice in $G \setminus H$ if no color is unused in $H$, that too when one of them is either the minimum or the maximum color in that span.
\end{lemma}

\begin{proof}
By Lemma \ref{fork3}, it is clear that if $c'$ is used at an edge in a fork in $H$ and reused thrice in $G \setminus H$ then, both of the colors $c'\pm 1$ have to be used at the rest two edges of the fork. Similarly, if $c'+1$ (or, $c'-1$) is used of at an edge of the fork and reused three times in $G \setminus H$ then $c'+2$ (or, $c'$) and $c'$ (or, $c'-2$)  have to be used at the other two edges of that fork, which is not possible as $c'-1$ (or, $c'+1$) and $c'$  have already been used at other two edges of that fork. So, the color $c'+2$ (or, $c'-2$) will remain unused in $H$, which makes $\lambda'_{1,2}(H) \geq 26$. So, if any of $c'\pm 1$ is reused thrice in $G \setminus H$ along with $c'$ then one of $c'\pm 2$ will be unused in $H$ and increase the span. However, observe that, if we can make sure that $c'+1=max$ or $c'-1=min$ then $c'+ 2$ or $c'-2$ does not belong to the span, where $max$ and $min$ be the maximum and the minimum colors in the span, respectively. This implies that when $c'+1=max$ (or, $c'-1=min$), both $c'$ and $c' + 1$ (or, $c'-1$) can be reused thrice each in $G \setminus H$. 
\end{proof}

\begin{theorem}
$\lambda'_{1,2}(T_8)\geq 26$.
\end{theorem}
\begin{proof}
From Lemma \ref{fork2} it can be said that, if no color is unused in $H$ then only one color used at an edge in a fork in $H$, can be reused three times in $G \setminus H$ unless that color is $min+1$ or $max-1$. There are four forks in $H$. So, two colors $min$ and $min+1$ from one fork having colors $\{min, min+1,min+2\}$, two colors $max-1$ and $max$ from another fork having colors $\{max-2,max-1,max\}$, one color each from the remaining two forks, can be reused thrice each in  $G \setminus H$. 
Hence, total six colors used in $H$ can be reused thrice each in $G \setminus H$, rest $20$ colors can be reused at most  twice. Total number of edges of $G \setminus H$ that can be colored using the colors used in $H$ is $(20 \times 2 + 6\times 3)=58$. But $60$  edges are there in $G \setminus H$. Hence, at least one more color is required.  
\end{proof}

\section{Conclusions}

Here we improve lower and upper bounds for $L(2,1)$-edge labeling of $T_3,T_4,T_6$ and $T_8$ using structural properties. An interesting problem will be to improve or introduce new bounds on those graphs for other values of $h$ and $k$ or to examine different variants of $L(h,k)$-labeling.




\end{document}